\documentclass[letter, 10pt, conference]{ieeeconf}      

\IEEEoverridecommandlockouts                              

\title{\LARGE \bf
Real-time Nonlinear Model Predictive Control using \\One-step Optimizations and Reachable Sets*
}

\author{Jan Olucak,$^{1}$ Walter Fichter,$^{1}$ and Torbjørn Cunis$^{1,2}$
\thanks{*The results presented in this paper have been achieved by the project "Astrone AI - AI support for high surface mobility of planetary research platforms (agile, autonomous, robust).", which has received funding from the German Federal Ministry for Economic Affairs and Energy (BMWi) under funding numbers “50 RA 2130C” supervised by the German Space Agency (DLR Raumfahrtagentur).}
\thanks{$^{1}$Jan Olucak, Walter Fichter and Torbjørn Cunis are with the Institute for Flight Mechanics and Control, University of Stuttgart
   70569 Stuttgart, Germany, 
   {\tt \{jan.olucak  | walter.fichter| torbjoern.cunis\} @ifr.uni-stuttgart.de}.} %
\thanks{$^{2}$Torbjørn Cunis is with the Department of Aerospace Engineering, University of Michigan, 
   Ann~Arbor, MI 48109, USA, {\tt tcunis@umich.edu}}%
}

\usepackage{graphicx}      
\usepackage{amsmath,amssymb, dsfont, cancel,mathrsfs}

\makeatletter
\let\NAT@parse\undefined
\makeatother

\let\ieeebibliography\thebibliography

\usepackage[numbers,sort&compress]{natbib}
\usepackage[colorlinks]{hyperref}

\renewcommand\thebibliography[1]{\ieeebibliography{#1}}

\newcommand{\trans}{^\mathsf{T}}

\def\equationautorefname~#1\null{(#1)\null}

\newcommand{\Algoref}[1]{Algorithm \ref{#1}}

\usepackage{algorithm}
\usepackage[noend]{algpseudocode}

\usepackage{tikz}
\usepackage{tikz-3dplot}
\tikzset{>=latex} 
\usepackage{pgfplotstable}
\usepackage{pgfplots}
\usepgfplotslibrary{groupplots}

\newlength\figH
\newlength\figW

\newtheorem{theo}{Theorem}
\newtheorem{lem}[theo]{Lemma}
\newtheorem{prop}[theo]{Proposition}

\newtheorem{defn}[theo]{Definition}
\newtheorem{assum}[theo]{Assumption}

\usepackage{float}
\usepackage{stfloats}

\begin{document}

\maketitle
\thispagestyle{empty}
\pagestyle{empty}

\begin{abstract}
Model predictive control allows solving complex control tasks with control and state constraints. However, an optimal control problem must be solved in real-time to predict the future system behavior, which is hardly possible on embedded hardware. To solve this problem, this paper proposes to compute a sequence of one-step optimizations aided by pre-computed inner approximations of reachable sets rather than solving the full-horizon optimal control problem at once. This feature can be used to virtually predict the future system behavior with a low computational footprint.
Proofs for recursive feasibility and for the sufficient conditions for asymptotic stability under mild assumptions are given. 
The presented approach is demonstrated in simulation for functional verification.
\end{abstract}

\section{Introduction}
\label{sec: intro}
The task of motion planning essentially involves planning the safe motion between an initial and final state \cite{la_valle_planning_2006}. In addition to control and state constraints, complex path constraints, mission objectives, and differential constraints, i.e., system dynamics must be considered. A frequently used method to solve such a complex problem is model predictive control (MPC) \cite{grune_nonlinear_2017}. Here, an optimal control problem is repeatedly solved that takes into account the aforementioned constraints. Typically, longer prediction horizons are necessary to ensure stability under path constraints.

A main problem for the use on embedded hardware is the computation time of the underlying optimization problem. Depending on the problem's size and complexity as well as the application, a near-optimal solution cannot be computed within one sampling period. Thus, to enable online use, methods are sought that enable a real-time  computation while guaranteeing feasibility and (asymptotic) stability.

In the literature, different approaches can be found that try to solve the real-time optimization problem. One approach are efficient low-level solvers for the underlying parameter optimization problems such as, for example, \cite{stella_simple_2017-1}, which was used in \cite{sathya_embedded_2018} for a collision avoidance scenario. Anefficient solver can drastically reduce the computation time, however, the feasibility of the problem must be guaranteed on the optimal control problem level. Hence, such a solver could be combined with appropriate formulations. 

To reduce computation time, warm-starting techniques can be used. For example in \cite{zeilinger_real-time_2014}, a warm-start procedure is proposed that allows real-time computation and ensures the feasibility of a linear robust MPC. Stability is guaranteed by a so-called Lyapunov constraint, a Lyapunov decrease condition.

Another approach is sub-optimal MPC, also known as time-distributed optimization \cite{liao-mcpherson_time-distributed_2020}. Instead of fully solving the underlying optimal control problem, only a limited number of iterations are performed. A well-known approach for sub-optimal MPC is the real-time iteration scheme (RTI) \cite{diehl_real-time_2005}. While sub-optimal MPC offers a real-time capable implementation, typical MPC stability guarantees do not hold anymore \cite{leung_computable_2021}. Instead, stability of sub-optimal MPC must be proven separately. Such a proof is given for the RTI scheme in\cite{diehl_nominal_2005}.

To guarantee stability and having a low computational effort, one-step ahead MPC (prediction horizon is one) is used, for example, in \cite{balau_one_2011,hermans_horizon-1_2013}. Both approaches rely on the results in \cite{lazar_flexible_2009} where so-called flexible Lyapunov functions are calculated by optimization to guarantee stability. While this method is very efficient for control affine systems and guarantees stability, it is less suitable for motion planning problems where a certain look-ahead is needed to avoid obstacles.

The approach in \cite{limon_robust_2003} uses pre-computed sequences of invariant sets to reduce the prediction horizon to a single step. However, {contraction (and hence convergence) is enforced by an additional constraint that requires an auxiliary optimization to be solved before evaluating the MPC feedback law.

To guarantee safety (e.g. collision-free trajectories), a low-computational footprint, feasibility, and asymptotic stability, a one-step MPC scheme employing pre-computed reachable sets is proposed. Thus, instead of solving potentially large-scale optimal control problems over the full horizon for the prediction, we are solving a sequence of smaller-sized one-step optimizations. Thereby, recursive feasibility and full-horizon constraint satisfaction are guaranteed by the reachable set. 

Reachable sets for MPC have previously been used, for example, by  \cite{alamo_robust_2005,bravo_robust_2006,schurmann_reachset_2018,skibik_feasibility_2022}, to ensure the robustness of the MPC.
There exist several techniques to compute reachable sets, including Hamilton--Jacobi-type reachability analysis \citep{lygeros_reachability_2004,bansal_hamilton-jacobi_2017}, set-propagation \citep{althoff_set_2021}, interval analysis \citep{meyer_interval_2021}, sampling-based methods \citep{liebenwein_sampling-based_2018} or storage functions \citep{cunis_viability_2021}.\\ In addition to the means of computation, the approaches differ in the description of the sets (e.g. zonotopes, polynomials, or grid). The suitability of a method depends on the particular application.

In this paper we are interested in the feasibility and real-time capability of the MPC rather than in computing a reachable set or an approximation.
Unlike previous MPC schemes with reachable sets, our approach does not require an additional contraction constraint but relies on the terminal penalty to achieve closed-loop asymptotic stability. Thus, we effectively obtain a horizon-one MPC formulation which is very similar to the original, full-horizon MPC.

The contribution of this paper is threefold. First, a one-step MPC that uses pre-computed reachable sets as a virtual prediction of the full-horizon system response is formulated. Second, we prove recursive feasibility and asymptotic stability under mild assumptions about the approximation of the reachable set. The third contribution is a method based on polynomial optimization to guarantee that sufficient conditions for asymptotic stability are met.

We demonstrate the usefulness of our approach in an illustrative example and show that the approach is significantly faster compared to the full-horizon formulation.

The remainder of this paper is organized as follows. The problem statement is given in \autoref{sec: ProbState}.  The one-step MPC and the corresponding ingredients are explained in  \autoref{sec: 1StepMPC}. In \autoref{sec: TheoInvest} the proofs for recursive feasibility and for asymptotic stability are given. Furthermore, a method is described to ensure the sufficient conditions for asymptotic stability are fulfilled. Numerical results are provided in \autoref{sec: NumRes}.

\section{Problem Statement}
\label{sec: ProbState}
We consider discrete-time dynamics which can be expressed as the difference equation
\begin{equation}
    x_{t+1} = f(x_t,u_t),
    \label{eq: dynDisc}
\end{equation}
where $t \in \mathcal{T} = \{0,1,\ldots, T\}$ with $T \in \mathbb{N}$, the vector field $f: \mathbb{R}^n \times \mathbb{R}^m \rightarrow \mathbb{R}^n$ is bounded in $\mathbb{R}^n$, Lipschitz continuous and describes the successor state vector $x_{t+1} \in \mathbb{R}^n$ that can be reached by the system starting in $x_t \in \mathbb{R}^n$ by applying input $u_t \in \mathbb{R}^m$.

Let $\mathcal{U} \subset \mathbb{R}^m$ be a closed and compact set that describes the admissible inputs. We denote the control sequence over horizon $T$ by $\mathbf{u} = \{u_0, \cdots , u_{T-1} \} \in \mathcal{U}^T$.
For $x_0 \in \mathbb{R}^n, t_0 \in \mathcal{T}$ and $\mathbf{u} \in \mathcal{U}^T$, the solution of \autoref{eq: dynDisc} after time $ t \in \mathcal{T}$ is denoted by $x(t,\mathbf{u},t_0,x_0)$.

The major objective of MPC is to find an optimal feedback sequence by prediction of \autoref{eq: dynDisc} over $\mathcal{T}$. In a basic MPC scheme, an optimal control problem is solved over a horizon $T \geq 2$ to predict the future response of the system \cite{grune_nonlinear_2017}.\\ The optimal control problem reads
\begin{subequations}
\label{eq: discTimeOCP}
\begin{alignat}{2}
&\!\min_{\mathbf{u} , x_1 ,\dots x_T}        &\qquad & x_T\trans Px_T + \sum_{t=0}^{T-1} W(x_t,u_t), \label{eq:optProb}\\
&\text{s.t.} &         & x_{t+1}  = f(x_t,u_t)  \,\,\,\,\,\,\,\,\,\,\,\,\,\,\,\,\,\,\, \forall t \in [0,T], \label{eq:constraint1}\\
&            &         &  x(T,\mathbf{u},t_0,x_0) \in \mathcal{X}_T, \label{eq:terminalSetCon}\\
&           &         &  x(t,\mathbf{u},t_0,x_0) \in \mathcal{X} \,\,\,\,\,\,\,\,\,\,\,\,\,\,\, \forall t \in [0,T], \label{eq:pathcon} \\
&           &         & u_t \in \mathcal{U} \,\,\,\,\,\,\,\,\,\,\,\,\,\,\,\,\,\,\,\,\,\,\,\,\,\,\,\,\,\,\,\,\,\,\,\,\,\,\,\,\,\,\, \forall t \in [0,T-1],\label{eq:Contcon}
\end{alignat}
\end{subequations}
where  $x \in \mathbb{R}^n$ is the state vector, $\mathcal{X} \subset \mathbb{R}^n$ describes the state constraint set, $\mathcal{X}_T  \subset \mathbb{R}^n$ is the terminal set, and $P \in \mathbb{R}^{n \times n}$ is a positive definite weight matrix of the terminal cost. The term $W(\cdot)$ describes the stage cost and reads 
\begin{equation}
    W(x,u) = x\trans Q x + u\trans R u,
\end{equation}
where $Q \in \mathbb{R}^{n \times n}$ and $R \in \mathbb{R}^{m \times m}$ are positive definite weight matrices. 
The set of all $\mathbf{u} \in \mathcal{U}^T$ that fulfill the constraints \autoref{eq:terminalSetCon} and \autoref{eq:pathcon} is denoted by $\mathcal{L}(x_0)$. The feasible set 
\begin{equation}
    \mathcal{F} = \{ x_0 \in \mathbb{R}^n \mid \mathcal{L}(x_0) \neq \emptyset \},
    \label{eq: feasSet}
\end{equation}
is the set of initial conditions that end in the terminal set.
Denote the optimal control sequence with respect to \autoref{eq: discTimeOCP} by $\mathbf{u}^*$. Then, the MPC feedback law reads
\begin{equation}
    \mathcal{C}(x(t)) = {u}^*(0).
\end{equation}

 This process is always restarted once a new measured state is available and solved over the full horizon until, e.g., a terminal set is reached or the system has converged.
 
 In general, the underlying optimization problem must be solved faster than the sampling rate of the system \cite{leung_computable_2021}. However, solving a full-horizon, potentially large-scale optimal control problem is hardly possible in real-time on embedded hardware. Therefore, an approximation method for the underlying optimal control problem is often sought.

To enable real-time solutions of the optimal control problem, we propose in this paper to solve a sequence of one-step optimizations using pre-computed reachable sets instead of solving the full-horizon problem repeatedly. This one-step MPC is explained in the next section, together with the underlying key ingredients.

\section{One-Step Model Predictive Control}
\label{sec: 1StepMPC}
In this section, the proposed one-step MPC is explained in detail. The core element of this approach lies in one-step optimizations in combination with reachable sets. Hence, we first give a short introduction to reachability analysis and explain, how these reachable sets aid in efficiently approximating full-horizon optimal control problems. Based on this information the one-step MPC is formulated, which uses the reachable set for the prediction.

\subsection{Reachable Sets and Feasibility}
A reachable set contains all states which are passed by trajectories at a certain time. Formally, the state-constrained reachable set reads
\begin{multline}
    \mathcal{R}_{[t_0, T]} = 
    \{x_0 \in \mathbb{R}^n \, | \, \exists \mathbf{u} \in \mathcal{U}^T, 
    x(T,\mathbf{u},t_0,x_0) \in \mathcal{X}_T \\ \wedge \forall t  \in [t_0, T], x(t,\mathbf{u},t_0,x_0) \in \mathcal{X}\}
    \label{eq: backreachset}
\end{multline}
for all $t_0 \in \{0, 1, \ldots, T\}$.
The reachable set describes the set of all initial states $\mathcal{X}_0 \in \mathbb{R}^n$ from which one can reach a given target set of terminal states $\mathcal{X}_T$  within $[t_0, T]$. 

There exists a relationship between the feasible set of \autoref{eq: discTimeOCP} and the reachable set of \autoref{eq: dynDisc}  \citep[Prop.~2]{cunis_viability_2021} as summarized in the following lemma.
\begin{lem}
Let $\mathcal{F}$ be the feasible set \autoref{eq: feasSet} of the optimal control problem \autoref{eq: discTimeOCP} and let $\mathcal{R}_{[0,T]}$ be the reachable set as defined in \autoref{eq: backreachset}; then $x \in \mathcal{F}$ if and only if $x \in \mathcal{R}_{[0,T]}$.
 \label{lem: FeasEqReach}
\end{lem}

Simply speaking, by computing the reachable set, with respect to the terminal set and constraints, one also computes the feasible set with respect to state and terminal constraints. 

In this paper, we are interested in the efficient solution of the underlying optimization of MPC aided by reachable sets, rather than the estimation of reachable sets. There exist different techniques to compute reachable sets (see \autoref{sec: intro}), where some methods provide inner and/or outer approximations. To guarantee feasibility of the optimization we require an inner approximate of the reachable set. To this extent, we make use of storage functions.

We assume $\mathcal{X}_T= \{x \in \mathbb{R}^n \mid l(x) \leq 0\}$ and $ \mathcal{X} = \{x \in \mathbb{R}^n \mid g(x) \leq 0\}$ for suitable constraint functions $g$ and $l$.

\begin{defn}
A function $V: \mathbb{N}_0 \times \mathbb{R}^n \rightarrow \mathbb{R}$ is a storage function for \autoref{eq: dynDisc} if and only if
\begin{align}
    \exists u \in \mathcal U, \,\, V(t+1,f(x,\mathbf{u}))-V(t,x) \leq 0,
\end{align}
for all $ (t,x) \in \mathcal{T} \times \mathbb{R}^n$.
\label{def: StorageFun}
\end{defn}
\begin{theo}
    Any storage function satisfying $V(t,x) \geq g(x)$ and $V(T,x) \geq l(x)$ for all $(t,x) \in \mathcal{T} \times \mathbb{R}^n$ provides an inner approximation of the reachable set, that is, 
\begin{equation}
    \label{eq: innerApprox}
    \{x_0 \in \mathbb{R}^n \mid V(0,x_0) \leq 0\} \subseteq  \mathcal{R}_{[0,T]}.
\end{equation}
\end{theo}
\begin{proof}
    Analogue to \citep[Proof of Theo.~14]{cunis_viability_2021}.
\end{proof} 

Such a storage function is guaranteed to exist as optimal value function of a dynamic programming problem \citep[Prop.~25]{cunis_viability_2021}.

For the remainder of this paper, it is assumed that a storage function for the reachable set is available on $\mathcal{T}$.

\subsection{One-step Optimization Algorithm}
\label{subsec: OneStepOpt}
In this subsection, we will formulate the one-step MPC scheme. As explained in \autoref{sec: ProbState}, the major computational burden in classical MPC comes from solving an optimal control problem over a potentially long horizon to predict the future system behavior. In particular, we need a long horizon to ensure feasibility under terminal constraints.

To solve this problem we propose to solve a sequence of one-step optimizations aided by reachable sets. 
\begin{figure}[h!]
    \centering
   \includegraphics{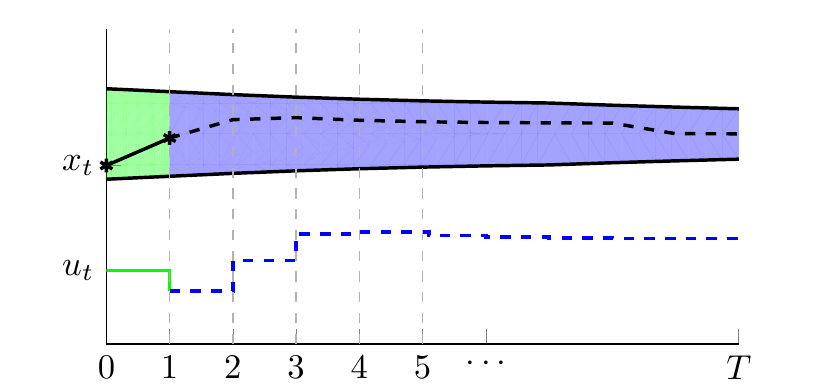}
    \caption{If the initial state is in the reachable set at time $t = 0$, then for each initial state there is a control sequence that brings the system into the terminal set. Only the first step of the optimization must be executed to be feasible (green area). The dashed lines indicate that a feasible trajectory can be recovered from the reachable set if one computes the subsequent steps.} 
    \label{fig:ComPredMPC}
\end{figure}
There is a prediction over $\mathcal{T}$ by the reachable set as indicated in \autoref{fig:ComPredMPC} illustrates this principle. This is the major difference to previous horizon-one MPC schemes
(see, e.g., \cite{balau_one_2011,hermans_horizon-1_2013}).

The one-step optimization problem reads
\begin{subequations}
\label{eq: 1stepOptMPC}
\begin{align}
&\!\min_{u ,x_+}        &\qquad & \alpha V(1,x_+) +  W(x_0,u), \label{eq:optProb1StepMPC}\\
&\text{subject to} &         & x_+ = f(x_0,u),\label{eq:1StepMPCconstraint1}\\
&                  &         & V(1,x_+) \leq 0,\label{eq:1StepMPCconstraint2}\\
&                  &         & u \in \mathcal{U},
\end{align}
\end{subequations}
where $(u, x_+) \in \mathbb{R}^{m+n}$ are the decision variables, $V(\cdot) \in \mathbb{R}$ is a storage function that describes the inner approximation of the reachable set, $x_0$ is the initial condition for the optimization, and $\alpha > 0$ is a weight factor for the terminal-cost term. This weight factor is theoretically analyzed in \autoref{sec: TheoInvest}.
The proposed one-step MPC can be viewed as a horizon-one MPC with the reachable set $\mathcal{R}_{[1,T]}$ as the terminal set and $\alpha V(1,\cdot)$ as terminal penalty; and neither the terminal set nor terminal penalty depend on the sampling time.

\begin{prop}
   The optimization problem \autoref{eq: 1stepOptMPC} is feasible for any $x_0 \in \mathbb{R}^n$ with $V(0,x_0) \leq 0$.
    \label{prop: Feas1Step}
\end{prop}
\begin{proof}
    Assume that $V(0,x_0) \leq 0$ for $x_0$. Per definition there exists $u \in \mathcal{U}$ such that $V(1,f(x_0,u)) \leq V(0,x_0) \leq 0$. Hence, $(u, f(x_0, u))$ is a feasible solution of \eqref{eq: 1stepOptMPC}.
\end{proof}

\Algoref{alg: Seq1StepMPC} shows how the MPC feedback law is calculated. Since the horizon is one, the optimal control sequence reduces to a single entry $u^* \in \mathcal{U}$.  $\hat{x}_{\mathcal{C}}(t)$ is the estimated state of the closed-loop system under the MPC feedback law. \Algoref{alg: Seq1StepMPC} is repeatedly executed until, e.g., a certain condition is fulfilled.  The optimal feedback-law $\mathcal{C}(x(t))$ is applied to the real plant $\hat{f}$, which is not necessarily the plant used in the one-step optimization, e.g., due to uncertainties or disturbances. We prove the sufficient conditions for asymptotic stability in the next section. This implies robustness to small disturbances by classical MPC results.

\begin{algorithm}[h!]
\caption{Calculate the MPC feedback law for current time step by one-step optimization given $x_0 = \hat{x}_{\mathcal{C}}(t) , W: \mathbb{R}^n \times \mathbb{R}^m \rightarrow \mathbb{R}_{\geq 0}, V: \mathbb{N}_{0} \times \mathbb{R}^n \rightarrow \mathbb{R}$}
\label{alg: Seq1StepMPC}
\begin{algorithmic}[1]
    \Require Initial condition $x_0$ satisfies $V(0,x_0) \leq 0$ 
        \State Solve \autoref{eq: 1stepOptMPC} to obtain $u^*$ 
        \State Set $\mathcal{C}(x(t)) = u^*$
        \State Apply $\mathcal{C}(x(t))$ to obtain $\hat{x}_{\mathcal{C}}(t+1) = \hat{f}(x(t),\mathcal{C}(x(t)))$
\end{algorithmic}
\end{algorithm}

\section{Theoretical Analysis}
\label{sec: TheoInvest}
In this section, we theoretically analyze the proposed horizon-one MPC with regard to the weight factor $\alpha$. We first prove recursive feasibility. Furthermore, we prove asymptotic stability of the MPC algorithm for a sufficiently large value of $\alpha$. 
At the end of this section, we propose a method to compute an $\alpha$, based on sum-of-squares (SOS) polynomials, for the system to be asymptotically stable.

From \autoref{subsec: OneStepOpt} we know that the problem is a horizon-one MPC. This allows us to use classical MPC theory after \citep{grune_nonlinear_2017}.
To show stability, we will prove that the requirements in \citep[Assum.~5.9]{grune_nonlinear_2017} for a Lyapunov function terminal cost are fulfilled by the proposed horizon-one MPC. The first requirement is that the terminal constraint set from \autoref{eq: 1stepOptMPC} is viable. The second is that the terminal cost from \autoref{eq: 1stepOptMPC} is a local control Lyapunov function.
We make the following assumptions.

\begin{assum}
The storage function $V(\cdot)$ in \autoref{eq: 1stepOptMPC} provides an non-empty inner approximation of the reachable set,
\begin{align}
  \mathcal{\tilde{R}}_{[t_0, T]} = \{x_0 \in \mathbb{R}^n \mid V(t_0,x_0) \leq 0\} \subseteq  \mathcal{R}_{[t_0,T]}, 
\end{align}
for all  $t_0 \in \{0,1,\ldots , T\}$.
Moreover, $V(1,x)$ is continuous in $x \in \mathbb R^n$ and $\mathcal{\tilde{R}}_{[1, T]}$ is compact.
\end{assum}

If one knows the exact reachable set $\mathcal{R}_{[0,T]}$, then if the terminal set in \autoref{eq: discTimeOCP} is viable, then $\mathcal{R}_{[0,T]}$ is viable. However, viability of the terminal set in \autoref{eq: discTimeOCP} does not automatically guarantee viability of an inner approximation. Thus, we need to encode a contraction into the calculation of the inner approximation. 

\begin{assum}
    The function $V(\cdot)$ satisfies $V(1,x) < V(2,x)$ for all $x \in \mathcal{\tilde{R}}_{[2,T]}\setminus\{0\}$.
    \label{assum: contraction}
\end{assum}

In order to prove that $V(\cdot)$ is a Lyapunov control function we need the following implication for dissipative functions.

\begin{lem}
    \cite[Lem.~11]{cunis_viability_2021}
    For all $(t_0,x_0) \in \mathcal T \times \mathbb R^n$,
    there exists a control $\mathbf{u} \in \mathcal U^T$ such that $ V(t,x(t,\mathbf u,t_0,x_0)) \leq V(0,x(0))$ for all $t > t_0$. In particular, for any $x \in \mathbb R^n$ there exists a control $u \in \mathcal U$ such that $V(2,f(x,u) \leq V(1,x)$.
    \label{lem: Dissipation}
\end{lem}

We will first prove recursive feasibility of \autoref{eq: 1stepOptMPC} in \Algoref{alg: Seq1StepMPC}. Then we will show that the horizon-one MPC terminal set is viable. Afterward, we will show that for a sufficiently large $\alpha$  the terminal cost is a local control Lyapunov function. The proof for asymptotic stability is given afterward. Finally, a method to estimate an $\alpha$ such that the sufficient conditions for asymptotic stability hold is provided in the end.

\subsection{Recursive Feasibility}
\begin{theo}
Equation \autoref{eq: 1stepOptMPC} in \Algoref{alg: Seq1StepMPC} is recursively feasible.
\label{theo: recFeas}
\end{theo}

\begin{proof}
Let $u^*$ be the solution of \autoref{eq: 1stepOptMPC}, that is, $V(1,x_+)\leq 0$. According to Lemma \autoref{lem: Dissipation}, there exists $ u_+ \in \mathcal{U}$ such that $V(2,f(x_+,u_+)) \leq V(1,x_+)$. By Assumption \autoref{assum: contraction}, it follows that
$V(1,f(x_+,u_+)) \leq V(2,f(x_+,u_+)) \leq V(1,x_+) \leq 0$. Hence, $u_+$ is a feasible solution of \autoref{eq: 1stepOptMPC} given $x_0 = x_+$.
\end{proof}

Now, knowing that \autoref{eq: 1stepOptMPC} in \Algoref{alg: Seq1StepMPC} is recursively feasible we can show that the horizon-one MPC terminal set is viable.

\begin{prop}
    The terminal set in the optimization \autoref{eq: 1stepOptMPC}, that is, $\mathcal{\tilde{R}}_{[1, T]} = \{x \in \mathbb{R}^n \mid V(1,x_+) \leq 0\}$ is viable.
    \label{prop: Viability}
\end{prop}

\begin{proof}
    We want to show that for all $x \in \mathcal{\tilde{R}}_{[1, T]}$ there exists an admissible control input $u \in \mathcal{U}$ such that $ f(x_0,u) \in \mathcal{\tilde{R}}_{[1, T]}$. If $x \in \mathcal{\tilde{R}}_{[1, T]}$, then $x \in \mathcal{\tilde{R}}_{[0, T]}$ by Assumption \autoref{assum: contraction}. Furthermore, by Theorem \autoref{theo: recFeas}, there exists $u \in \mathcal{U}$ such that $f(x,u) \in \mathcal{\tilde{R}}_{[1, T]}$.
\end{proof}

\subsection{Asymptotic Stability}
To prove stability of the horizon-one MPC we have to incorporate the weight factor $\alpha$ into the analysis.
\begin{prop}
     There exists an $ \alpha_0$ such that \begin{align}
         \exists u \in \mathcal{U}, \,\,W(x,u) \leq \alpha_0(V(1,x)-V(1,f(x,u))),
        \label{eq: costFunLyappF}
    \end{align}
    for all $x \in \mathcal{\tilde{R}}_{[1,T]}$ is fulfilled, i.e., the horizon-one MPC terminal cost is a local control Lyapunov function.
    \label{prop: ContLyap}
\end{prop}
   
\begin{proof}
    Take $x \in \mathcal{\tilde{R}}_{[1,T]}$;
     from Lemma \ref{lem: Dissipation} and  Assumption \ref{assum: contraction}, there exists a $u \in \mathcal{U}$ such that $V(1,f(x,u)) < V(2,f(x,u)) \leq V(1,x)$, that is, $V(1,x)-V(1,f(x,u^*)) > 0$. 
     Since $\mathcal{\tilde{R}}_{[1, T]}$ and $\mathcal{U}$ are bounded, $W(\cdot)$ has an upper limit; moreover, since $V(\cdot)$ and $W(\cdot)$ are continuous functions, \autoref{eq: costFunLyappF} is satisfied if $\alpha_0$ is sufficiently large.
\end{proof}

\begin{theo}
    There exists an $\alpha_0$ such that the horizon-one MPC is asymptotically stable on $\mathcal{\tilde{R}}_{[0,T]}$ for all $\alpha \geq \alpha_0$.
\end{theo}

\begin{proof}
    By Proposition \autoref{prop: Viability} the terminal set is viable. If Proposition \ref{prop: ContLyap} with $\alpha_0$ holds, then the assumptions in \citep[Assum.~5.9]{grune_nonlinear_2017} for Lyapunov function terminal cost are fulfilled for all $\alpha \geq \alpha_0$. Asymptotic stability follows then from  \citep[Theo.~5.13]{grune_nonlinear_2017}, where the domain of stability contains $\mathcal{\tilde{R}}_{[0,T]}$ by Proposition~\autoref{prop: Feas1Step}.
\end{proof}

From the above theoretical analysis, a sufficiently large $\alpha_0$ is needed to guarantee asymptotic stability. We will provide a method to estimate such an $\alpha_0$ in the next section.

\subsection{$\alpha$-Weight Estimation}
In this section we provide an estimation method based on SOS programming \cite{parrilo_semidefinite_2003} to find an $\alpha_0$ for the system to be asymptotically stable. A polynomial $q$ is SOS if and only if $q = \sum_{k=1}^m q_i^2$ exists, with $q_i \in \mathbb{R}[x]$ where $\mathbb{R}[x]$ denotes the ring of polynomials in $x$. 
In case that $q(x) \in \Sigma[x]$, this implies that $q(x) \geq 0 $ for all $x \in \mathbb{R}^n$,
where $\Sigma[x] \subset \mathbb{R}[x]$ denotes the set of SOS polynomials.
 
A SOS program reads \cite{seiler_quasiconvex_2010}
\begin{subequations}
\label{eq: SOSprog}
\begin{align}
&\!\min_{d \in \mathbb{R}^r}        &\quad &\phi(d)\\
&\text{subject to} &         & a_k(x,d) \in \Sigma[x], \,\, k = 1,\dots ,N,
\end{align}
\end{subequations}
where $\phi(\cdot) \in  \mathbb{R}$ is a linear cost function, $d \in \mathbb{R}^r$ are the decision variables and $a_k(\cdot)$ are given polynomials in $x,d$. $x$ are free variables for the optimization. SOS programs can be converted into semidefinite-programs \cite{parrilo_semidefinite_2003,seiler_quasiconvex_2010}. There exist several toolboxes such as 
sosopt \cite{seiler_sosopt_2010} to do this conversion.

To synthesize a storage function  $V(\cdot)$ as defined in Definition \ref{def: StorageFun}, we use the approach from \cite[Sec. 5]{cunis_viability_2021} where SOS programming is used. By computing a storage function this way, one obtains a polynomial storage function $V(\cdot)$ and the set of corresponding viable inputs, \cite[Eq. 22--23]{cunis_viability_2021}, denoted by $h_\mathcal{U} \in \mathbb{R}[t,x]^m$. The function $h_\mathcal{U}(t,x)$ provides a control input such that $f(x,h_\mathcal{U}(t,x))  \in \mathcal{\tilde{R}}_{[t+1,T]}$ if $x \in \mathcal{\tilde{R}}_{[t,T]}$ for for all $(x,t)$. 

Now to find $\alpha_0$ such that \autoref{eq: costFunLyappF} is fulfilled, we need
\begin{multline}
    \{x \in \mathbb{R}^n | -V(1,x) \geq 0 \} \subseteq \\ \{x \in \mathbb{R}^n | \alpha_0 (V(1,x) -V(1,f(x,u)))-W(x,u) \geq 0 \},
\end{multline}
where $u = h_\mathcal{U}(1,x)$.
 This can be cast into an SOS program by applying the Positivstellensatz \cite{parrilo_semidefinite_2003}
\begin{subequations}   
\label{eq: SOSprogMinAlpha}
\begin{align}
&\!\min_{\alpha_0 \in \mathbb{R}}        & & \alpha_0,\\
&\text{s.t.} &         &  s(x) \in \Sigma[x],\\
&                  &         & s(x)V(1,x) + \alpha_0 \Delta V  \nonumber\\
&                  &         &  \quad - W(x,h_\mathcal{U}(1,x)) \in \Sigma[x],
\end{align}
\end{subequations}
where $\Delta V = V(1,x)-V(1,f_p(x,h_\mathcal{U}(1,x))$, $s(x)$ is an SOS-multiplier, $f_p \in \mathbb{R}[x]$ is a polynomial approximation of the vector field $f$ (e.g. by Taylor approximation). It is important to mention that the polynomial approximation is only used to compute the storage function and sufficient $\alpha$ weight. In the actual one-step optimization \autoref{eq: 1stepOptMPC} the vector field $f$ is used.  

\begin{figure*}[tp]
  \includegraphics{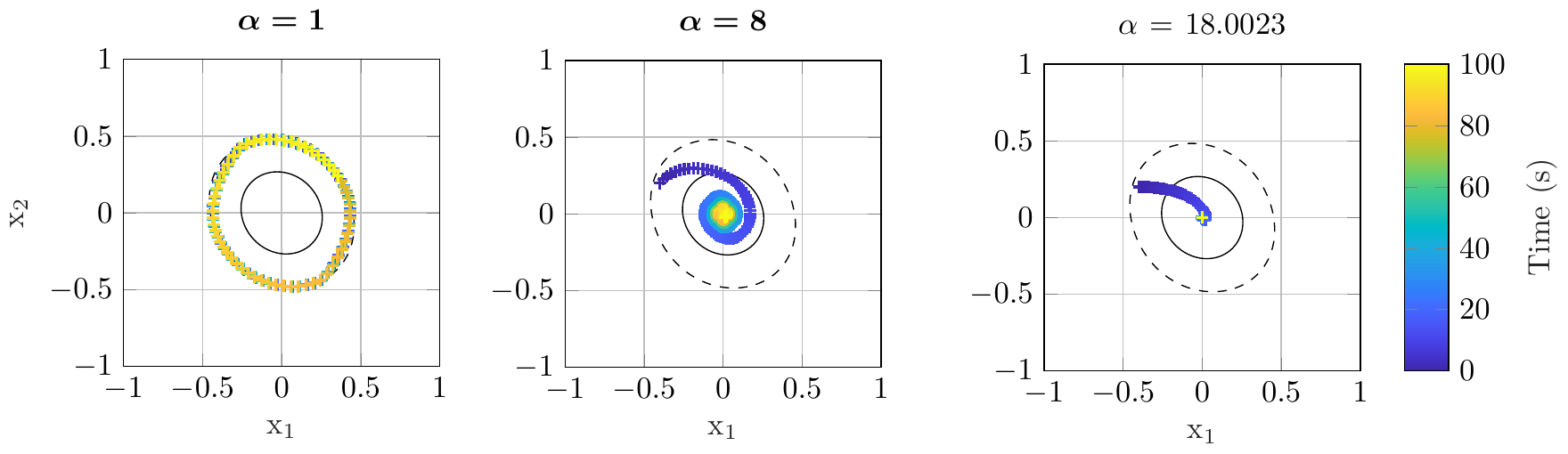}
  \caption{Closed-loop system trajectories under one-step MPC feedback for different $\alpha$ weights.} 
   \label{fig:PhasePlots}
\end{figure*}

\section{Numerical Results}
\label{sec: NumRes}

The main objective of this section lies in the verification of the proposed approach in simulation. To solve the optimal control problem, CasADi \citep{andersson_casadi_2019} and IPOPT \cite{wachter_implementation_2006} are used.\\ The reachable set is approximated on  $\mathcal{T} = [0,10]$. The one-step optimization uses Runge-Kutta 4/5 discretization and the time step is set to 0.1s. All computations are executed with Matlab on a personal computer with Windows 10, an AMD Ryzen 9 5950X 16-Core Processor 3.40 GHz on a single-core with 16 GB RAM.

\subsection{Forced Van-der-Pol Oscillator}
The following problem is from \cite[Sec. 6]{cunis_viability_2021}. The forced Van-der-Pol oscillator dynamics in continuous time read
\begin{equation}
    \dot{x}= \begin{bmatrix}
    x_2 \\ (1-x_1^2)x_2-x_1+u
    \end{bmatrix}.
\end{equation}
Constraints on the input are imposed as
\begin{equation}
    -1 \leq u \leq 1.
\end{equation}
Furthermore, a state constraint is imposed that reads
\begin{equation}
    g(x) = 1-x_1^2-3x_2^2 \geq 0.
    \label{eq: pathConVan}
\end{equation}
The terminal constraint reads,
\begin{equation}
    l(x) = 1-x\trans Px \geq 0,
    \label{eq: terminalSetVan}
\end{equation}
where $P \in \mathbb{R}^{n \times n}$
\begin{equation}
    P = \begin{bmatrix}
         6.4314 &   0.4580 \\
    0.4580   & 5.8227
    \end{bmatrix}. \nonumber
\end{equation}

The constraints are encoded in the reachable set, i.e., in the storage function $V(\cdot)$.  The stage cost reads $W(x,u) = x_1^2+x_2^2+u^2$. The system equilibrium point is at $x = (0,0)\trans$ and is contained in the terminal set. We used \autoref{eq: SOSprogMinAlpha} to compute an $\alpha_0 = 18.0023$ for \autoref{eq: costFunLyappF} to be satisfied. The problem is initialized at $x_0 = (-0.4, 0.2)\trans$.

\subsection{Results}
In \autoref{fig:PhasePlots} the closed-loop system trajectories for different $\alpha$ values are depicted. For $\alpha = 1$ (left) it is obvious that the system is in a limit cycle and does not converge to the equilibrium point. In the mid figure, one can see that the system trajectory seemingly converges to the equilibrium point for $\alpha = 8$.

\begin{figure}[H]
  \includegraphics{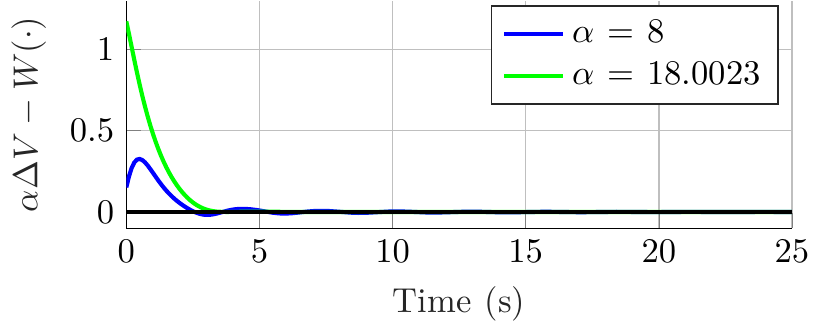}
  \caption{Difference between the left and right-hand side of \autoref{eq: costFunLyappF} evaluated along closed-loop system trajectories for different $\alpha$ weights; negative values correspond to violating the inequality.}
   \label{fig: Eq11}
\end{figure}

\begin{figure}[H]
  \includegraphics{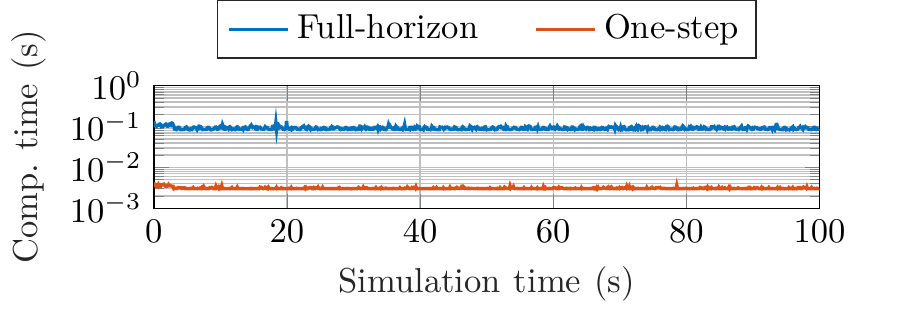}
  \caption{Comparison of computation time over simulation time for the full-horizon problem and for the one-step optimization with $\alpha_0 = 18.0023$.}  
   \label{fig:CompTime}
\end{figure}

In the third figure (right) the system trajectory for  $\alpha = 18.0023$ converges to the origin. In \autoref{fig: Eq11}, the difference in \autoref{eq: costFunLyappF} is evaluated along the trajectories. Since for $\alpha = 1$ the system does not converge to the equilibrium only the trajectories for $\alpha \in \{8, 18.0023\}$ are depicted. One can see that for $\alpha = 8$ the sufficient conditions for asymptotic stability are violated whenever the zero line is crossed. For $\alpha = 18.0023$, the sufficient conditions are fulfilled, as expected. 

In \autoref{fig:CompTime} a computation time comparison between the full-horizon MPC and the one-step approach is depicted. One sees that the computation time of the one-step optimization is very low. The worst-case computation time to solve the full-horizon problem is about 137 ms whereas the worst-case computation time for one optimization of the one-step approach is about 4 ms.

\section{Conclusions}
\label{sec: Con}
In this paper, we demonstrate that it is possible to efficiently solve state-constrained MPC problems via one-step optimizations aided by pre-computed reachable sets. A theoretical analysis proves recursive feasibility and sufficient conditions for asymptotic stability of the proposed one-step MPC. In particular, it is shown that asymptotic stability can be ensured with a soft-constraint only.
In order for the proofs to hold true, we made some mild assumptions about the reachable set along the prediction horizon. The proposed method is demonstrated in a numerical example. 
Here, the computation time was more than 30-times faster than a full-horizon MPC feedback.

\addtolength{\textheight}{-12cm}   




\bibliographystyle{IEEEtran}
\bibliography{CDC20232}

\begin{thebibliography}{10}
\providecommand{\url}[1]{#1}
\csname url@rmstyle\endcsname
\providecommand{\newblock}{\relax}
\providecommand{\bibinfo}[2]{#2}
\providecommand\BIBentrySTDinterwordspacing{\spaceskip=0pt\relax}
\providecommand\BIBentryALTinterwordstretchfactor{4}
\providecommand\BIBentryALTinterwordspacing{\spaceskip=\fontdimen2\font plus
\BIBentryALTinterwordstretchfactor\fontdimen3\font minus
  \fontdimen4\font\relax}
\providecommand\BIBforeignlanguage[2]{{%
\expandafter\ifx\csname l@#1\endcsname\relax
\typeout{** WARNING: IEEEtran.bst: No hyphenation pattern has been}%
\typeout{** loaded for the language `#1'. Using the pattern for}%
\typeout{** the default language instead.}%
\else
\language=\csname l@#1\endcsname
\fi
#2}}

\bibitem{la_valle_planning_2006}
S.~M. La~Valle, \emph{Planning {Algorithms}}.\hskip 1em plus 0.5em minus
  0.4em\relax Cambridge University Press, 2006.

\bibitem{grune_nonlinear_2017}
L.~Grüne and J.~Pannek, \emph{\BIBforeignlanguage{en}{Nonlinear {Model}
  {Predictive} {Control}}}, ser. Communications and {Control}
  {Engineering}.\hskip 1em plus 0.5em minus 0.4em\relax Springer International
  Publishing, 2017.

\bibitem{stella_simple_2017-1}
L.~Stella, A.~Themelis, P.~Sopasakis, and P.~Patrinos, ``A simple and efficient
  algorithm for nonlinear model predictive control,'' in \emph{2017 {IEEE} 56th
  {Annual} {Conference} on {Decision} and {Control}}, Dec. 2017, pp.
  1939--1944.

\bibitem{sathya_embedded_2018}
A.~Sathya, P.~Sopasakis, R.~Van~Parys, A.~Themelis, G.~Pipeleers, and
  P.~Patrinos, ``Embedded nonlinear model predictive control for obstacle
  avoidance using {PANOC},'' in \emph{2018 {European} {Control} {Conference}},
  June 2018, pp. 1523--1528.

\bibitem{zeilinger_real-time_2014}
M.~N. Zeilinger, D.~M. Raimondo, A.~Domahidi, M.~Morari, and C.~N. Jones,
  ``\BIBforeignlanguage{en}{On real-time robust model predictive control},''
  \emph{\BIBforeignlanguage{en}{Automatica}}, vol.~50, no.~3, pp. 683--694,
  Mar. 2014.

\bibitem{liao-mcpherson_time-distributed_2020}
D.~Liao-McPherson, M.~M. Nicotra, and I.~Kolmanovsky,
  ``\BIBforeignlanguage{en}{Time-distributed optimization for real-time model
  predictive control: {Stability}, robustness, and constraint satisfaction},''
  \emph{\BIBforeignlanguage{en}{Automatica}}, vol. 117, July 2020.

\bibitem{diehl_real-time_2005}
M.~Diehl, H.~G. Bock, and J.~P. Schlöder, ``\BIBforeignlanguage{en}{A
  {Real}-{Time} {Iteration} {Scheme} for {Nonlinear} {Optimization} in
  {Optimal} {Feedback} {Control}},'' \emph{\BIBforeignlanguage{en}{SIAM Journal
  on Control and Optimization}}, vol.~43, no.~5, pp. 1714--1736, Jan. 2005.

\bibitem{leung_computable_2021}
J.~Leung, D.~Liao-McPherson, and I.~V. Kolmanovsky, ``\BIBforeignlanguage{en}{A
  {Computable} {Plant}-{Optimizer} {Region} of {Attraction} {Estimate} for
  {Time}-distributed {Linear} {Model} {Predictive} {Control}},'' in
  \emph{\BIBforeignlanguage{en}{2021 {American} {Control} {Conference}}}, New
  Orleans, LA, USA, May 2021, pp. 3384--3391.

\bibitem{diehl_nominal_2005}
M.~Diehl, R.~Findeisen, F.~Allgöwer, H.~G. Bock, and J.~P. Schlöder,
  ``\BIBforeignlanguage{en}{Nominal stability of real-time iteration scheme for
  nonlinear model predictive control},'' \emph{\BIBforeignlanguage{en}{IEE
  Proceedings - Control Theory and Applications}}, vol. 152, no.~3, pp.
  296--308, May 2005, publisher: IET Digital Library.

\bibitem{balau_one_2011}
A.~Balau and C.~Lazar, ``One {Step} {Ahead} {MPC} for an {Automotive} {Control}
  {Application},'' in \emph{2011 {Second} {Eastern} {European} {Regional}
  {Conference} on the {Engineering} of {Computer} {Based} {Systems}}, Sept.
  2011, pp. 61--70.

\bibitem{hermans_horizon-1_2013}
R.~M. Hermans, M.~Lazar, I.~V. Kolmanovsky, and S.~Di~Cairano, ``Horizon-1
  {Predictive} {Control} of {Automotive} {Electromagnetic} {Actuators},''
  \emph{IEEE Transactions on Control Systems Technology}, vol.~21, no.~5, pp.
  1652--1665, Sept. 2013.

\bibitem{lazar_flexible_2009}
M.~Lazar, ``Flexible control {Lyapunov} functions,'' in \emph{2009 {American}
  {Control} {Conference}}, June 2009, pp. 102--107.

\bibitem{limon_robust_2003}
D.~Limon, T.~Alamo, and E.~Camacho, ``Robust {MPC} control based on a
  contractive sequence of sets,'' in \emph{42nd {IEEE} {International}
  {Conference} on {Decision} and {Control}}, vol.~4, Dec. 2003, pp. 3706--3711
  vol.4.

\bibitem{alamo_robust_2005}
T.~Alamo, D.~Limon, E.~Camacho, and J.~Bravo, ``\BIBforeignlanguage{en}{Robust
  {MPC} of constrained nonlinear systems based on interval arithmetic},''
  \emph{\BIBforeignlanguage{en}{IEE Proceedings - Control Theory and
  Applications}}, vol. 152, no.~3, pp. 325--332, May 2005.

\bibitem{bravo_robust_2006}
J.~M. Bravo, T.~Alamo, and E.~F. Camacho, ``\BIBforeignlanguage{en}{Robust
  {MPC} of constrained discrete-time nonlinear systems based on approximated
  reachable sets},'' \emph{\BIBforeignlanguage{en}{Automatica}}, vol.~42,
  no.~10, pp. 1745--1751, Oct. 2006.

\bibitem{schurmann_reachset_2018}
B.~Schürmann, N.~Kochdumper, and M.~Althoff, ``Reachset {Model} {Predictive}
  {Control} for {Disturbed} {Nonlinear} {Systems},'' in \emph{2018 {IEEE}
  {Conference} on {Decision} and {Control}}, Dec. 2018, pp. 3463--3470.

\bibitem{skibik_feasibility_2022}
T.~Skibik, D.~Liao-McPherson, T.~Cunis, I.~Kolmanovsky, and M.~M. Nicotra, ``A
  {Feasibility} {Governor} for {Enlarging} the {Region} of {Attraction} of
  {Linear} {Model} {Predictive} {Controllers},'' \emph{IEEE Transactions on
  Automatic Control}, vol.~67, no.~10, pp. 5501--5508, Oct. 2022.

\bibitem{lygeros_reachability_2004}
J.~Lygeros, ``\BIBforeignlanguage{en}{On reachability and minimum cost optimal
  control},'' \emph{\BIBforeignlanguage{en}{Automatica}}, vol.~40, no.~6, pp.
  917--927, June 2004.

\bibitem{bansal_hamilton-jacobi_2017}
S.~Bansal, M.~Chen, S.~Herbert, and C.~J. Tomlin,
  ``\BIBforeignlanguage{en}{Hamilton-{Jacobi} reachability: {A} brief overview
  and recent advances},'' in \emph{\BIBforeignlanguage{en}{2017 {IEEE} 56th
  {Annual} {Conference} on {Decision} and {Control}}}, Melbourne, Australia,
  Dec. 2017, pp. 2242--2253.

\bibitem{althoff_set_2021}
M.~Althoff, G.~Frehse, and A.~Girard, ``Set {Propagation} {Techniques} for
  {Reachability} {Analysis},'' \emph{Annual Review of Control, Robotics, and
  Autonomous Systems}, vol.~4, no.~1, pp. 369--395, 2021.

\bibitem{meyer_interval_2021}
P.-J. Meyer, A.~Devonport, and M.~Arcak, \emph{\BIBforeignlanguage{en}{Interval
  {Reachability} {Analysis}: {Bounding} {Trajectories} of {Uncertain} {Systems}
  with {Boxes} for {Control} and {Verification}}}, ser. {SpringerBriefs} in
  {Electrical} and {Computer} {Engineering}.\hskip 1em plus 0.5em minus
  0.4em\relax Springer International Publishing, 2021.

\bibitem{liebenwein_sampling-based_2018}
L.~Liebenwein, C.~Baykal, I.~Gilitschenski, S.~Karaman, and D.~Rus,
  ``\BIBforeignlanguage{en}{Sampling-{Based} {Approximation} {Algorithms} for
  {Reachability} {Analysis} with {Provable} {Guarantees}},'' in
  \emph{\BIBforeignlanguage{en}{Robotics: {Science} and {Systems}
  {XIV}}}.\hskip 1em plus 0.5em minus 0.4em\relax Robotics: Science and Systems
  Foundation, June 2018.

\bibitem{cunis_viability_2021}
T.~Cunis and I.~Kolmanovsky, ``\BIBforeignlanguage{en}{Viability, viscosity,
  and storage functions in model-predictive control with terminal
  constraints},'' \emph{\BIBforeignlanguage{en}{Automatica}}, vol. 131, Sept.
  2021.

\bibitem{parrilo_semidefinite_2003}
P.~A. Parrilo, ``\BIBforeignlanguage{en}{Semidefinite programming relaxations
  for semialgebraic problems},'' \emph{\BIBforeignlanguage{en}{Mathematical
  Programming}}, vol.~96, no.~2, pp. 293--320, May 2003.

\bibitem{seiler_quasiconvex_2010}
P.~Seiler and G.~J. Balas, ``\BIBforeignlanguage{en}{Quasiconvex sum-of-squares
  programming},'' in \emph{\BIBforeignlanguage{en}{49th {IEEE} {Conference} on
  {Decision} and {Control}}}.\hskip 1em plus 0.5em minus 0.4em\relax Atlanta,
  GA, USA: IEEE, Dec. 2010, pp. 3337--3342.

\bibitem{seiler_sosopt_2010}
\BIBentryALTinterwordspacing
P.~Seiler, ``{SOSOPT}: {A} toolbox for polynomial optimization,'' 2010.
  [Online]. Available: \url{https://dept.aem.umn.edu/~AerospaceControl/}
\BIBentrySTDinterwordspacing

\bibitem{andersson_casadi_2019}
J.~A.~E. Andersson, J.~Gillis, G.~Horn, J.~B. Rawlings, and M.~Diehl,
  ``\BIBforeignlanguage{en}{{CasADi}: a software framework for nonlinear
  optimization and optimal control},''
  \emph{\BIBforeignlanguage{en}{Mathematical Programming Computation}},
  vol.~11, no.~1, pp. 1--36, Mar. 2019.

\bibitem{wachter_implementation_2006}
A.~Wächter and L.~T. Biegler, ``\BIBforeignlanguage{en}{On the implementation
  of an interior-point filter line-search algorithm for large-scale nonlinear
  programming},'' \emph{\BIBforeignlanguage{en}{Mathematical Programming}},
  vol. 106, no.~1, pp. 25--57, Mar. 2006.

\end{thebibliography}

\end{document}